\documentclass{revtex4}
\usepackage{amsmath,amsthm}
\usepackage{rotating}
\usepackage{multirow}

\usepackage{graphicx}

\usepackage{amsfonts}
\newtheorem{theorem}{Theorem}[section]
\newtheorem{lemma}[theorem]{Lemma}
\newtheorem{proposition}[theorem]{Proposition}

\theoremstyle{remark}
\newtheorem{remark}[theorem]{Remark}

\theoremstyle{definition}
\newtheorem{definition}[theorem]{Definition}

\theoremstyle{example}
\newtheorem{example}[theorem]{Example}

\theoremstyle{notation}

\newcommand{\bra}[1]{\langle#1|}
\newcommand{\ket}[1]{|#1\rangle}

\begin{document}

\title{Quantum probabilities as Dempster-Shafer probabilities in the lattice of subspaces}            
\author{A. Vourdas}
\affiliation{Department of Computing,\\
University of Bradford, \\
Bradford BD7 1DP, United Kingdom}

\begin{abstract}
The orthocomplemented modular lattice of subspaces ${\cal L}[H(d)]$, of a quantum system with $d$-dimensional Hilbert space $H(d)$, is considered.
A generalized additivity relation which holds for Kolmogorov probabilities,
is violated by quantum probabilities in the full lattice ${\cal L}[H(d)]$ (it is only valid within the Boolean subalgebras of ${\cal L}[H(d)]$).
This suggests the use of more general (than Kolmogorov) probability theories, and here the
Dempster-Shafer probability theory is adopted.
An operator ${\mathfrak D}(H_1, H_2)$, which quantifies deviations from Kolmogorov probability theory is introduced,
and it is shown to be intimately related to the commutator of the projectors ${\mathfrak P}(H_1), {\mathfrak P}(H_2)$,
to the subspaces $H_1,H_2$.
As an application, it is shown that the proof of CHSH inequalities for a system of two spin $1/2$ particles, 
is valid for Kolmogorov probabilities, but it is not valid for Dempster-Shafer probabilities.
The violation of these inequalities in experiments, supports the interpretation of quantum probabilities as Dempster-Shafer probabilities.
\end{abstract}

\maketitle

\section{Introduction}

Probability theory needs for its axioms, the concepts of conjuction, disjunction and negation, and in this sense it is tacitly defined
with respect to a lattice. 
Kolmogorov's probability theory is intimately connected to Boolean algebras.
It is defined on a powerset $2^\Omega$ of some set $\Omega$, and 
subset ($\subseteq$), intersection ($\cap$) union ($\cup$), complement (${\overline A}=\Omega -A$), are the logical connectives 
partial order ($\prec$), conjuction ($\wedge$), disjunction ($\vee$), negation ($\neg$), correspondingly.
A basic relation of Kolmogorov probabilities $q(A)$ is that
\begin{eqnarray}\label{I1}
\delta (A,B)=0;\;\;\;\;\delta(A,B)=q(A\vee B)-q(A)-q(B)+q(A\wedge B);\;\;\;A,B\in 2^\Omega.
\end{eqnarray}
This is a `generalized additivity relation' (the term `generalized' refers to the $q(A\wedge B)$ for non-exclusive events).

Quantum logic \cite{LO1,LO2,LO3,LO4,LO5,LO6} is based on the orthomodular lattice of closed subspaces of a Hilbert space, which
has various Boolean algebras as sublattices. 
We show that Eq.(\ref{I1}) is valid for quantum probabilities only within a Boolean subalgebra, but it
is violated in the full orthomodular lattice.
This leads to more general (than Kolmogorov) probability theories, which violate the additivity relation of 
Eq.(\ref{I1})\cite{F1,F2,F3,F4}. Non-additive probabilities are used in cases where
there is added value in a coalition.
In everyday language this is described with the expression 
`the whole is greater than the sum of its parts'. 
One such theory is the Dempster-Shafer theory, which has been used extensively in Artificial Intelligence, Operations Research, Economics, etc \cite{DS1,DS2,DS3,DS4}.

The original problem that Dempster considered \cite{DS1} was that of a sample space $X$ and a multivalued map 
$\Gamma$ from $X$ to another sample space $\Omega$.
He then carried Kolmogorov probabilities on subsets of $X$, into lower and upper probabilities on subsets of $\Omega$.
For lower probabilities $\delta (A,B)\ge 0$ and for upper probabilities $\delta (A,B)\le 0$. 
The need for lower and upper probabilities is intimately related to the multivaluedness of the map $\Gamma$.
The Dempster multivaluedness is similar to the passage from classical physics to quantum physics,
where classical quantities become operators, and we have the problem of ordering in products of these operators
(see below).

We introduce an operator ${\mathfrak D}(H_1, H_2)$ (where 
$H_1,H_2$ are subspaces of the Hilbert space of the system $H(d)$), 
which is analogous to $\delta (A,B)$. We show that this operator is related to
the commutator of the projectors ${\mathfrak P}(H_1), {\mathfrak P}(H_2)$ into $H_1,H_2$.
(proposition \ref{pro} below).
Also let $\rho $ be a density matrix, $p(H_i|\rho)= {\rm Tr}[\rho {\mathfrak P}(H_i)]$ and
${\mathfrak d}(H_1, H_2|\rho)= {\rm Tr}[\rho {\mathfrak D}(H_1,H_2)]$
(this is analogous to $\delta (A,B)$).
The $p(H_1|\rho)$, $p(H_2|\rho)$ can be viewed as lower or upper Dempster-Shafer probabilities,
according to whether ${\mathfrak d}(H_1, H_2|\rho)\ge 0$ or ${\mathfrak d}(H_1, H_2|\rho)\le 0$, correspondingly.
Therefore the ${\mathfrak D}(H_1, H_2)$ characterizes the nature of a pair of quantum probabilities $p(H_1|\rho)$, $p(H_2|\rho)$.

One application of these ideas is in the general area of Bell inequalities and contextuality\cite{C0,C1},
which has been studied extensively for a long time\cite{C2,C3,C4,C5,C6,C7,C8,C9,C10,C11,C12,C13,C14}.
In order to prove inequalities that involve quantum probabilities, 
it is important to understand the properties of these probabilities.
We show that the CHSH (Clauser, Horne, Shimony and Holt \cite{C2}) inequalities,
which are proved using the properties of Kolmogorov probabilities, 
do not hold for Dempster-Shafer probabilities.
Their violation in experiments supports the interpretation of quantum probabilities as Dempster-Shafer probabilities.

General probabilistic theories have been used in quantum mechanics by various authors\cite{G00,G0,G1,G2,G3,G4,G5}.
Operational approaches and convex geometry methods have been studied in \cite{CO1,CO2,CO3,CO4,CO5}.
Test spaces have been studied in \cite{TE1,TE2}.
Fuzzy phase spaces have been studied in \cite{FF1,FF2,FF3}.
Category theory methods have been studied in \cite{CA1,CA2}.
In this paper we show that the use of Dempster-Shafer probability theory alleviates the difficulties that
appear when we use Kolmogorov's theory for quantum probabilities.
An example of these difficulties is that they lead to various Bell-like inequalities, which are violated by experiment.

In section II we introduce the operator ${\mathfrak D}(H_1, H_2)$ as a quantum mechanical analogue to $\delta (A,B)$,
and we discuss its relationship with the commutator $[{\mathfrak P}(H_1), {\mathfrak P}(H_2)]$. 
In section III we state briefly the properties of the Dempster-Shafer (upper and lower) probabilities,  
and discuss their use as quantum probabilities.
In section IV we study CHSH (Clauser, Horne, Shimony and Holt \cite{C1}) inequalities.
We show that their proof is valid for Kolmogorov probabilities but it is invalid for Dempster-Shafer probabilities.
These inequalities are violated in experiments, and this supports the view that quantum probabilities are Dempster-Shafer probabilities. 

We conclude in section V with a discussion of our results. 

\section{Preliminaries}
We consider a quantum system with position and momenta in ${\mathbb Z}(d)$ (the integers modulo $d$), described by 
a $d$-dimensional Hilbert space $H(d)$.
In this space we consider the orthonormal basis of `position states' $\ket{X;n}$, which we represent with the vectors
$\ket{X;0}=(0,...,0,1)^T$, $\ket{X;1}=(0,...,1,0)^T$, etc.
We also consider the basis of momentum states
\begin{eqnarray}
\ket{P;m}=d^{-1/2}\sum _m\omega (mn)\ket{X;n};\;\;\;\;\omega(\alpha )=\exp \left(\frac{i2\pi \alpha}{d}\right );\;\;\;\;m,n, \alpha\in {\mathbb Z}(d)
\end{eqnarray}
and the displacement operators
\begin{eqnarray}
D(\alpha , \beta)=Z^{\alpha}X^{\beta}\omega (-2^{-1}\alpha \beta);\;\;\;\;
Z=\sum _m\omega (m)\ket{X;m}\bra{X;m};\;\;\;\;
X=\sum _m \ket{X;m+1}\bra{X;m}
\end{eqnarray}
Some aspects of the formalism of finite quantum systems, is slightly different in the cases of odd or even $d$.
For example the factor $2^{-1}$ above, is an element of ${\mathbb Z}(d)$, and it exists for odd $d$. 
Below, in the formulas that use the displacement operators, we assume that the dimension $d$ is an odd integer. 

Acting with $D(\alpha , \beta)$ on a (normalized) fiducial vector 
\begin{eqnarray}
\ket {f}=\sum _mf_m\ket{X;m},
\end{eqnarray}
which should not be a position state, or a momentum state, 
we get the following $d^2$ states which we call coherent states\cite{COH,COH1}
\begin{eqnarray}\label{coh}
\ket{C;\alpha, \beta}=D(\alpha , \beta)\ket{f};\;\;\;\;\alpha , \beta \in {\mathbb Z}(d).
\end{eqnarray}
The $X,P,C$ in the notation are not variables, but they simply indicate position states, momentum states and coherent states.
The coherent states obey the resolution of the identity equation:
\begin{eqnarray}\label{coh}
\frac{1}{d}\sum_{\alpha, \beta}\ket{C;\alpha, \beta}\ket{C;\alpha, \beta}={\bf 1}.
\end{eqnarray}
For later use we give the formula
\begin{eqnarray}\label{la}
\lambda (\alpha, \beta ;\gamma, \delta)=\bra{C;\alpha, \beta}C;\gamma, \delta\rangle=
\omega [2^{-1}(\alpha \beta+\gamma \delta)-\alpha \delta]\sum _{n}f_{n+\delta -\beta}^*f_n\omega [n(\gamma -\alpha )]
\end{eqnarray}

We next consider a spin $1/2$ particle described with the Hilbert space $H(2)$.
$S_x,S_y,S_z$ are spin operators in the directions $x,y,z$, correspondingly.
The vectors $\ket{\frac{1}{2},\frac{1}{2}}=(1,0)^T$ and $\ket{\frac{1}{2},-\frac{1}{2}}=(0,1)^T$ represent spin up and spin down states in the $z$-direction.
Also
\begin{eqnarray}
&&S_x=\frac{1}{2}\left(
\begin{array}{cc}
0&1\\
1&0\\
\end{array}
\right );\;\;\;\;\;
S_y=\frac{1}{2}\left(
\begin{array}{cc}
0&-i\\
i&0\\
\end{array}
\right );\;\;\;\;\;
S_z=\frac{1}{2}\left(
\begin{array}{cc}
1&0\\
0&-1\\
\end{array}
\right )
\end{eqnarray}
We express $S_x$ in terms of projectors as
\begin{eqnarray}\label{3}
&&S_x=\frac{1}{2}\Pi (x,1)-\frac{1}{2}\Pi(x,0);\;\;\;\;
\Pi (x,1)=\frac{1}{2}
\left(
\begin{array}{cc}
1&1\\
1&1\\
\end{array}
\right );\;\;\;\;\;\; 
\Pi (x,0)=\frac{1}{2}\left(
\begin{array}{cc}
1&-1\\
-1&1\\
\end{array}
\right )={\bf 1}_2-\Pi (x,1).
\end{eqnarray}

With the $SU(2)$ rotation
\begin{eqnarray}\label{3a}
U(a,b)=
\left(
\begin{array}{cc}
a&b\\
-b^*&a^*\\
\end{array}
\right );\;\;\;\;\;\; |a|^2+|b|^2=1
\end{eqnarray}
we get 
\begin{eqnarray}\label{3b}
&&S_{a,b}=U(a,b)S_x[U(a,b)]^{\dagger}=\frac{1}{2}\Pi(a,b;1)-\frac{1}{2}\Pi(a,b;0)\nonumber\\
&&\Pi (a,b;1)=\frac{1}{2}
\left(
\begin{array}{cc}
|a+b|^2&a^2-b^2\\
(a^*)^2-(b^*)^2&|a-b|^2\\
\end{array}
\right );\;\;\;\;\;\; 
\Pi (a,b;0)=\frac{1}{2}\left(
\begin{array}{cc}
|a-b|^2&-a^2+b^2\\
-(a^*)^2+(b^*)^2&|a+b|^2\\
\end{array}
\right )
\end{eqnarray}
Clearly $S_{0,0}=S_x$. 

For a system of two spin $1/2$ particles, described with the Hilbert space
$H(2)\otimes H(2)=H(4)$, we use analogous notation.

\section{The orthocomplemented modular lattice ${\cal L}[H(d)]$ and the generalized additivity operator ${\mathfrak D}(H_1, H_2)$}

The closed subspaces of a Hilbert space are themselves Hilbert spaces, and  they form an orthomodular lattice,
which has been studied extensively after the work of Birkhoff and von Neumann on quantum logic\cite{LO1,LO2,LO3,LO4,LO5}.
General references on lattices are \cite {BIR,BIR1,BIR2}, and on orthomodular lattices \cite{O1,O2,O3,O4}.
Let ${\cal L}[H(d)]$ be the
orthomodular lattice of subspaces of the $d$-dimensional Hilbert space $H(d)$.
In this case all subspaces are closed, and 
for this reason we ommit the term 'closed'.
Another consequence of the fact that our Hilbert spaces are finite dimensional, is that the orthomodular lattice ${\cal L}[H(d)]$,
is modular. Orthomodularity is weaker concept than orthocomplemented modularity and the ${\cal L}[H(d)]$ is orthocomplemented modular lattice.

The logical connectives in the lattice ${\cal L}[H(d)]$ are defined as follows.
If $H_1,H_2$ are subspaces of $H(d)$, then
\begin{eqnarray}
H_1\wedge H_2=H_1\cap H_2;\;\;\;\;\;H_1\vee H_2={\rm span}(H_1\cup H_2)
\end{eqnarray} 
We use the notation $\cal O$ and $\cal I$ for the smallest and greatest elements in the lattice. Then
${\cal O}=H(0)$ is the zero-dimensional subspace that contains only the zero vector, and ${\cal I}=H(d)$.
$H_1\prec H_2$ means that $H_1$ is a subspace of $H_2$.
The orthocomplement  
$H_1^{\bot}$ is a subspace of $H(d)$, orthogonal to $H_1$, such that
\begin{eqnarray}\label{690}
H_1\wedge H_1^{\bot}={\cal O};\;\;\;\;\;H_1\vee H_1^{\bot}={\cal I}.
\end{eqnarray} 
We denote as ${\mathfrak P}(H_1)$ the projector to the subspace $H_1$.
In the following lemma we give without proof, some elementary properties which are needed later:
\begin{lemma}
\mbox{}
\begin{itemize}
\item[(1)]
\begin{eqnarray}\label{e50}
&&{\mathfrak P}(H_1\wedge H_2){\mathfrak P}(H_1)={\mathfrak P}(H_1){\mathfrak P}(H_1\wedge H_2)=
{\mathfrak P}(H_1\wedge H_2)\nonumber\\
&&{\mathfrak P}(H_1\vee H_2){\mathfrak P}(H_1)={\mathfrak P}(H_1){\mathfrak P}(H_1\vee H_2)={\mathfrak P}(H_1)
\end{eqnarray} 
\item[(2)]
\begin{eqnarray}\label{e5}
{\mathfrak P}(H_1){\mathfrak P}(H_1^{\bot})=0;\;\;\;\;\;{\mathfrak P}(H_1)+{\mathfrak P}(H_1^{\bot})={\bf 1};\;\;\;\;
{\mathfrak P}(H_1\wedge H_2){\mathfrak P}(H_1^{\bot}\wedge H_3)=0.
\end{eqnarray} 
\item[(3)]
If ${\mathfrak P}(H_1){\mathfrak P}(H_2)=0$ then ${\mathfrak P}(H_2){\mathfrak P}(H_1)=0$, $H_1\wedge H_2={\cal O}$ and
${\mathfrak P}(H_1)+{\mathfrak P}(H_2)={\mathfrak P}(H_1\vee H_2)$. The $H_1\wedge H_2={\cal O}$ does not necessarlily imply the
${\mathfrak P}(H_1){\mathfrak P}(H_2)=0$.
The ${\mathfrak P}(H_1){\mathfrak P}(H_2)=0$ is equivalent  to $H_1\prec H_2^{\bot}$, and it is denoted as $H_1\bot H_2$.
\end{itemize}
\end{lemma}

\begin{definition}
$H_1$ commutes with $H_2$ (the standard notation for this is $H_1{\cal C} H_2$) if
\begin{eqnarray}\label{45}
H_1=(H_1\wedge H_2)\vee (H_1\wedge H_2^{\bot})
\end{eqnarray}
\end{definition}
It is easily seen that $H_1 \prec H_2$ implies that $H_1{\cal C} H_2$. 
Also, since ${\cal L}[H(d)]$ is an orthomodular lattice,  if $H_1{\cal C} H_2$, then $H_2{\cal C} H_1$ and also $H_1{\cal C} H_2^{\bot}$,
$H_1^{\bot}{\cal C} H_2$. Furthermore, if $H_1{\cal C} H_2$ then
\begin{eqnarray}\label{AB1}
(H_1\wedge H_2)\vee (H_1\wedge H_2^{\bot})\vee (H_1^{\bot}\wedge H_2)\vee (H_1^{\bot}\wedge H_2^{\bot})
={\cal I},
\end{eqnarray}

\begin{proposition}\label{pro}
In the lattice ${\cal L}[H(d)]$, let
\begin{eqnarray}\label{32}
{\mathfrak D}(H_1, H_2)={\mathfrak P}(H_1\vee H_2)+{\mathfrak P}(H_1\wedge H_2)-{\mathfrak P}(H_1)-{\mathfrak P}(H_2).
\end{eqnarray} 
Then
\begin{itemize}
\item[(1)]
The fact that ${\cal L}[H(d)]$ is a modular lattice implies that ${\rm Tr}[{\mathfrak D}(H_1, H_2)]=0$.
\item[(2)]
The ${\mathfrak P}(H_1\wedge H_2)$ is in general different from ${\mathfrak P}(H_1){\mathfrak P}(H_2)$ and their difference is given by
\begin{eqnarray}\label{e30}
{\mathfrak P}(H_1\wedge H_2)-{\mathfrak P}(H_1){\mathfrak P}(H_2)={\mathfrak P}(H_1){\mathfrak D}(H_1, H_2)=
{\mathfrak D}(H_1, H_2){\mathfrak P}(H_2)
\end{eqnarray}
\item[(3)]
The commutator $[{\mathfrak P}(H_1),{\mathfrak P}(H_2)]$ is related to ${\mathfrak D}(H_1, H_2)$, through the relation:
\begin{eqnarray}\label{e3}
[{\mathfrak P}(H_1),{\mathfrak P}(H_2)]={\mathfrak D}(H_1, H_2)[{\mathfrak P}(H_1)-{\mathfrak P}(H_2)]
=-[{\mathfrak P}(H_1)-{\mathfrak P}(H_2)]{\mathfrak D}(H_1, H_2).
\end{eqnarray}
\item[(4)]
The following are equivalent:
\begin{itemize}
\item[(a)]
${\mathfrak D}(H_1, H_2)=0$.
\item[(b)]
$[{\mathfrak P}(H_1),{\mathfrak P}(H_2)]=0$.
\item[(c)]
${\mathfrak P}(H_1\wedge H_2)={\mathfrak P}(H_1){\mathfrak P}(H_2)$
\item[(d)]
$H_1{\cal C} H_2$
\end{itemize}
\item[(5)]
\begin{eqnarray}
[{\mathfrak P}(H_1),{\mathfrak P}(H_2)]=-[{\mathfrak P}(H_1),{\mathfrak D}(H_1, H_2)].
\end{eqnarray}
\item[(6)]
\begin{eqnarray}\label{32a}
{\mathfrak D}(H_1^{\bot}, H_2^{\bot})=-{\mathfrak D}(H_1, H_2).
\end{eqnarray} 
\item[(7)]
If $H_1\bot H_2$ or $H_1\prec H_2$, then ${\mathfrak D}(H_1, H_2)=0$. 
\end{itemize}
\end{proposition}
\begin{proof}
\mbox{}
\begin{itemize}
\item[(1)]
Let ${\rm dim}(H_1)$ be the dimension of $H_1$, and ${\mathfrak h}(H_1)$ the height of the element
$H_1$ in the lattice ${\cal L}[H(d)]$. Then
\begin{eqnarray}
{\mathfrak h}(H_1)={\rm dim}(H_1)={\rm Tr}[{\mathfrak P}(H_1)].
\end{eqnarray}
${\cal L}[H(d)]$ is a modular lattice and therefore the height obeys the realtion (e.g., \cite{BIR}, p.41)
\begin{eqnarray}
{\mathfrak h}(H_1\vee H_2)+{\mathfrak h}(H_1\wedge H_2)-{\mathfrak h}(H_1)-{\mathfrak h}(H_2)=0.
\end{eqnarray}
From this follows that 
\begin{eqnarray}
{\rm Tr}[{\mathfrak D}(H_1, H_2)]={\rm Tr}[{\mathfrak P}(H_1\vee H_2)]+{\rm Tr}[{\mathfrak P}(H_1\wedge H_2)]-{\rm Tr}[{\mathfrak P}(H_1)]-{\rm Tr}[{\mathfrak P}(H_2)]=0.
\end{eqnarray}

\item[(2)]
In order to prove this,
we multiply Eq.(\ref{32}) on the left hand side with ${\mathfrak P}(H_1)$, using Eq.(\ref{e50}).
We also multiply Eq.(\ref{32}) on the right hand side with ${\mathfrak P}(H_2)$, using Eq.(\ref{e50}).
\item[(3)]
This is proved using Eq.(\ref{e30}).
\item[(4)]
We prove the four parts of the `loop' $a\rightarrow b\rightarrow c\rightarrow d\rightarrow a$.

\begin{itemize}

\item[(i)]
If ${\mathfrak D}(H_1, H_2)=0$ then $[{\mathfrak P}(H_1),{\mathfrak P}(H_2)]=0$, according to Eq.(\ref{e3}).

\item[(ii)]
Let $[{\mathfrak P}(H_1),{\mathfrak P}(H_2)]=0$. From this follows that ${\mathfrak P}(H_1){\mathfrak P}(H_2)$ is a projector.
The ${\mathfrak P}(H_1){\mathfrak P}(H_2)$ projects into a space with vectors which belong to both $H_1$ and $H_2$.
Therefore these vectors belong to $H_1\cap H_2$, and this proves that ${\mathfrak P}(H_1){\mathfrak P}(H_2)$ is the same projector as
${\mathfrak P}(H_1\wedge H_2)$.

\item[(iii)]
We assume that ${\mathfrak P}(H_1\wedge H_2)={\mathfrak P}(H_1){\mathfrak P}(H_2)$.
Since the ${\mathfrak P}(H_1\wedge H_2)$ is symmetric (i.e., ${\mathfrak P}(H_1\wedge H_2)={\mathfrak P}(H_2\wedge H_1)$),
it follows that $[{\mathfrak P}(H_1),{\mathfrak P}(H_2)]=0$. From this we prove that
$[{\mathfrak P}(H_1),{\mathfrak P}(H_2^{\bot})]=0$, and 
therefore ${\mathfrak P}(H_1\wedge H_2^{\bot})={\mathfrak P}(H_1){\mathfrak P}(H_2^{\bot})$.

From ${\mathfrak P}(H_1\wedge H_2){\mathfrak P}(H_1\wedge H_2^{\bot})=0$, it follows that
\begin{eqnarray}
{\mathfrak P}[(H_1\wedge H_2)\vee (H_1\wedge H_2^{\bot})]
&=&{\mathfrak P}(H_1\wedge H_2)+{\mathfrak P}(H_1\wedge H_2^{\bot})\nonumber\\&=&{\mathfrak P}(H_1){\mathfrak P}(H_2)+{\mathfrak P}(H_1){\mathfrak P}(H_2^{\bot})=
{\mathfrak P}(H_1)
\end{eqnarray} 
Therefore $H_1=(H_1\wedge H_2)\vee (H_1\wedge H_2^{\bot})$ which proves that $H_1{\cal C} H_2 $.

\item[(iv)]
If $H_1{\cal C} H_2 $, then Eq.(\ref{45}), and the fact that ${\mathfrak P}(H_1\wedge H_2){\mathfrak P}(H_1\wedge H_2^{\bot})=0$, leads to
\begin{eqnarray}\label{71}
{\mathfrak P}(H_1)={\mathfrak P}(H_1\wedge H_2)+{\mathfrak P}(H_1\wedge H_2^{\bot}).
\end{eqnarray} 
In a similar way we get
\begin{eqnarray}\label{71a}
{\mathfrak P}(H_2)={\mathfrak P}(H_1\wedge H_2)+{\mathfrak P}(H_1^{\bot}\wedge H_2).
\end{eqnarray} 

We note that
\begin{eqnarray}\label{333}
{\mathfrak P}(A_i){\mathfrak P}(A_j)=0;\;\;\;
A_1=H_1\wedge H_2;\;\;\;A_2=H_1^{\bot}\wedge H_2;\;\;\;
A_3=H_1\wedge H_2^{\bot};\;\;\;
A_4=H_1^{\bot}\wedge H_2^{\bot},
\end{eqnarray}
and therefore Eq.(\ref{AB1}) gives
\begin{eqnarray}
{\mathfrak P}(H_1\wedge H_2)+{\mathfrak P}(H_1\wedge H_2^{\bot})+{\mathfrak P}(H_1^{\bot}\wedge H_2)+{\mathfrak P}(H_1^{\bot}\wedge H_2^{\bot})
={\bf 1}.
\end{eqnarray} 
From this follows that
\begin{eqnarray}\label{71b}
{\mathfrak P}(H_1\wedge H_2)+{\mathfrak P}(H_1\wedge H_2^{\bot})+{\mathfrak P}(H_1^{\bot}\wedge H_2)={\bf 1}-{\mathfrak P}(H_1^{\bot}\wedge H_2^{\bot})
={\mathfrak P}(H_1\vee H_2).
\end{eqnarray} 
We now add Eqs(\ref{71}), (\ref{71a}), taking into account Eq.(\ref{71b}), and we prove that ${\mathfrak D}(H_1, H_2)=0$.

\end{itemize}

\item[(5)]
The proof of this is straightforward (using Eq.(\ref{e50})).
\item[(6)]
The proof of this is straightforward (using Eq.(\ref{e5})).
\item[(7)]
The proof of this is straightforward.
\end{itemize}
\end{proof}

\begin{remark}
The operator ${\mathfrak D}(H_1,H_2)$ is analogous to $\delta (A,B)$ for Kolmogorov probabilities. Therefore
the ${\mathfrak D}(H_1,H_2)$, which in general is non-zero, quantifies how different quantum probabilities are from Kolmogorov probabilities.
Eq.(\ref{e3}) shows that the ${\mathfrak D}(H_1,H_2)$ 
is also a measure of non-commutativity. 
\end{remark}

Let $\rho$ is a density matrix of a system described by the space $H(d)$.
An important theorem by Gleason \cite{GL} shows that quantum probabilities associated with the projectors ${\mathfrak P}(H_1)$
are given by
$p(H_1|\rho)={\rm Tr}[\rho {\mathfrak P}(H_1)]$ (in spaces with dimension greater than $2$). Then
\begin{eqnarray}
p(H_1^{\bot}|\rho)=1-p(H_1|\rho).
\end{eqnarray}
Let
\begin{eqnarray}\label{327}
{\mathfrak d}(H_1, H_2|\rho)= {\rm Tr}[\rho {\mathfrak D}(H_1,H_2)]=p(H_1\vee H_2|\rho)+p(H_1\wedge H_2|\rho)-p(H_1|\rho)-p(H_2|\rho).
\end{eqnarray}
Then $H_1{\cal C} H_2 $ or equivalently $[{\mathfrak P}(H_1), {\mathfrak P}(H_2)]=0$, if and only if ${\mathfrak d}(H_1, H_2|\rho)=0$, for all density matrices.
For a particular density matrix
it may be that ${\mathfrak d}(H_1, H_2|\rho)=0$, although ${\mathfrak D}(H_1, H_2)$ is non-zero.

Let $\lambda _i$ and $\ket{v_i}$ (with $i=1,...,d$) be the eigenvalues and eigenvectors of ${\mathfrak D}(H_1, H_2)$. 
We order the eigenvalues  as $\lambda_1\le ...\le\lambda _d$. 
Then
\begin{eqnarray}
&&{\rm Tr}[{\mathfrak D}(H_1, H_2)]=\sum _i \lambda _i=0;\;\;\;\;\lambda_1\le 0 \le \lambda_{d}\nonumber\\
&&{\mathfrak d}(H_1, H_2|\rho)=\sum \lambda _i\bra{v_i}\rho \ket{v_i};\;\;\;\;\lambda _1\le {\mathfrak d}(H_1, H_2|\rho) \le \lambda _d.
\end{eqnarray}
For various density matrices, the ${\mathfrak d}(H_1, H_2|\rho)$ takes both positive and negative values.
For example:
\begin{eqnarray}
{\mathfrak d}(H_1, H_2|\ket{v_1}\bra{v_1})=\lambda _1 \le 0;\;\;\;\;
{\mathfrak d}(H_1, H_2|\ket{v_d}\bra{v_d})=\lambda _d \ge 0;\;\;\;\;
{\mathfrak d}(H_1, H_2|\frac{1}{d}{\bf 1})=0.
\end{eqnarray}

\begin{example}
In the space $H(3)$ we consider the one-dimensional subspaces
\begin{eqnarray}
H_1=\left \{a (1,0,0)^T\right \};\;\;\;\;H_2=\left \{a (1,1,0)^T\right \}
\end{eqnarray}
Here we give the general  vector that belongs to each of these spaces.
In this case
\begin{eqnarray}
H_1\vee H_2=\left \{(a,b,0)^T\right \};\;\;\;\;H_1\wedge H_2={\cal O}
\end{eqnarray}
The corresponding projectors are
\begin{eqnarray}
&&{\mathfrak P}(H_1)=\left(
\begin{array}{ccc}
1&0&0\\
0&0&0\\
0&0&0\\
\end{array}
\right );\;\;\;\;\;
{\mathfrak P}(H_2)=\frac{1}{2}\left(
\begin{array}{ccc}
1&1&0\\
1&1&0\\
0&0&0\\
\end{array}
\right );\;\;\;\;\;
{\mathfrak P}(H_1\vee H_2)=\left(
\begin{array}{ccc}
1&0&0\\
0&1&0\\
0&0&0\\
\end{array}
\right )
\end{eqnarray}
In this example, the $H_1,H_2$ do not commute. We find that 
\begin{eqnarray}
{\mathfrak D}(H_1, H_2)=\frac{1}{2}\left(
\begin{array}{ccc}
-1&-1&0\\
-1&1&0\\
0&0&0\\
\end{array}
\right );\;\;\;\;\;
[{\mathfrak P}(H_1),{\mathfrak P}(H_2)]=\frac{1}{2}\left(
\begin{array}{ccc}
0&1&0\\
-1&0&0\\
0&0&0\\
\end{array}
\right ),
\end{eqnarray}
and we can confirm Eq.(\ref{e3}).
\end{example}
\begin{remark}
For non-commuting projectors ${\mathfrak P}(H_1),{\mathfrak P}(H_2)$,
ref\cite{HA} has considered the set of states $\ket{s}$ in $H(d)$ for which
$[{\mathfrak P}(H_1),{\mathfrak P}(H_2)]\ket{s}=0$.
They form a subspace of $H(d)$ which we denote as ${\cal H}[{\mathfrak P}(H_1),{\mathfrak P}(H_2)]$.
More generally let $A,B$ be the non-commuting observables
\begin{eqnarray}
A=\sum _i\lambda _i {\mathfrak P}_{Ai};\;\;\;\;\;\;\;B=\sum _j\mu _j {\mathfrak P}_{Bj},
\end{eqnarray}
where $\lambda _i, \mu_j$ are their eigenvalues, and ${\mathfrak P}_{Ai}, {\mathfrak P}_{Bj}$ are their eigenprojectors.
The space ${\cal H}(A,B)$ contains all states $\ket{s}$ in $H(d)$ for which
$[{\mathfrak P}_{Ai},{\mathfrak P}_{Bj}]\ket{s}=0$ for all $i,j$.
Although $A,B$ do not commute, they are `weakly compatible', i.e., compatible with respect to the states in ${\cal H}(A,B)$
(ref.\cite{HA} uses the term `relative compatibility').
It will be interesting to extend these ideas in our context, 
i.e., for ${\mathfrak D}(H_1, H_2)\ne 0$, to consider the set ${\cal R}$ of density matrices  for which
${\mathfrak d}(H_1, H_2|\rho)={\rm Tr}[{\mathfrak D}(H_1, H_2)\rho]=0$.
Then ${\mathfrak P}(H_1),{\mathfrak P}(H_2)$ can be associated to Kolmogorov probabilities,
only for density matrices in the set ${\cal R}$.
We do not study further this idea in this paper.
\end{remark}

\subsection{Boolean algebras associated with orthonormal bases}

If $B$ is a sublattice of ${\cal L}[H(d)]$ and $H_1{\cal C} H_2$ for any two elements $H_1,H_2 \in B$, then 
$B$ is a Boolean subalgebra of ${\cal L}[H(d)]$ (the negation $\neg H_1$ is the $H_1^{\bot}$)\cite{BIR}. Maximal Boolean algebras in ${\cal L}[H(d)]$
are sometimes called Boolean blocks. 
In order to construct one, we consider a basis of $d$ vectors in $H(d)$, which are orthogonal to each other.
We then consider the $e$-dimensional space spanned by $e$ of these vectors. There are 
\begin{eqnarray}
n_d(e)=\left (
\begin{array}{c}
d\\
e\\
\end{array}
\right );\;\;\;\;\;0\le e\le d
\end{eqnarray}
such spaces. Therefore the total number of spaces in the Boolean algebra is $\sum _en_d(e)=2^d$.
Acting with a unitary transformation on a Boolean algebra $B$, we get another Boolean algebra.
${\cal L}[H(d)]$ is the union of its Boolean algebras, which are `pasted'
with rules discussed in \cite{O1,O2,O3,O4}.

For any two elements $H_1,H_2$ in a Boolean algebra, ${\mathfrak D}(H_1, H_2)=0$ and ${\mathfrak d}(H_1, H_2|\rho )=0$.
In this case the corresponding quantum probabilities $p(H_1|\rho)$, $p(H_2|\rho)$, are Kolmogorov probabilities.

Boolean algebras within ${\cal L}[H(d)]$ play the role of quantum contexts in the sense of Kochen and Specker\cite{C1}.
Examples of Boolean algebras within ${\cal L}[H(4)]$ are given later.

\subsection{Projectors associated to coherent states}

We consider the one-dinesional space $H_{\alpha,\beta}$ that contains the coherent state $\ket{C;\alpha,\beta}$, defined in Eq.(\ref{coh}).
The projector to this space is 
\begin{eqnarray}
{\mathfrak P}(H_{\alpha,\beta})=\ket{C;\alpha,\beta}\bra{C;\alpha,\beta};\;\;\;\;\;
\frac{1}{d}\sum _{\alpha,\beta}{\mathfrak P}(H_{\alpha,\beta})={\bf 1}.
\end{eqnarray}
For $(\alpha,\beta) \ne (\gamma, \delta)$ we get $H_{\alpha,\beta}\wedge H_{\gamma, \delta}={\cal O}$.
Also $H_{\alpha,\beta}\vee H_{\gamma, \delta}$ is the two-dimensional space that contains the vectors
$\kappa \ket{C;\alpha,\beta}+\mu \ket{C;\gamma,\delta}$, and the corresponding projector is
\begin{eqnarray}
&&{\mathfrak P}(H_{\alpha,\beta}\vee H_{\gamma, \delta})=\ket{C;\alpha,\beta}\bra{C;\alpha,\beta}+\ket{s}\bra{s}\nonumber\\
&&\ket{s}=(1-|\lambda (\alpha, \beta ;\gamma, \delta) |^2)^{-1/2}[\ket{C;\gamma, \delta}-\lambda (\alpha, \beta ;\gamma, \delta)\ket{C;\alpha,\beta}]
\nonumber\\
&&\bra{C;\alpha,\beta}s\rangle=0
\end{eqnarray}
where $\lambda (\alpha, \beta ;\gamma, \delta)$ has been given in Eq.(\ref{la}). We note that
\begin{eqnarray}
&&{\mathfrak P}(H_{\alpha,\beta}){\mathfrak P}(H_{\gamma, \delta})=\lambda (\alpha, \beta ;\gamma, \delta)\ket{C;\alpha,\beta}\bra{C;\gamma,\delta}
\nonumber\\
&&{\mathfrak P}(H_{\gamma, \delta}){\mathfrak P}(H_{\alpha,\beta})=[\lambda (\alpha, \beta ;\gamma, \delta)]^*\ket{C;\gamma,\delta}\bra{C;\alpha,\beta}\nonumber\\
&&|\lambda(\alpha, \beta ;\gamma, \delta)|^2={\rm Tr}[{\mathfrak P}(H_{\alpha,\beta}){\mathfrak P}(H_{\gamma, \delta})]
\end{eqnarray}
Therefore
\begin{eqnarray}
&&{\mathfrak P}(H_{\alpha,\beta}\vee H_{\gamma, \delta})=\frac{1}{1-|\lambda(\alpha, \beta ;\gamma, \delta)|^2}
\left[{\mathfrak P}(H_{\alpha,\beta})+{\mathfrak P}(H_{\gamma, \delta})
-{\mathfrak P}(H_{\alpha,\beta}){\mathfrak P}(H_{\gamma, \delta})-
{\mathfrak P}(H_{\gamma, \delta}){\mathfrak P}(H_{\alpha,\beta})\right ]\nonumber\\
\end{eqnarray}
and
\begin{eqnarray}
{\mathfrak D}(H_{\alpha,\beta},H_{\gamma, \delta})&=&\frac{|\lambda(\alpha, \beta ;\gamma, \delta)|^2}{1-|\lambda(\alpha, \beta ;\gamma, \delta)|^2}\left[{\mathfrak P}(H_{\alpha,\beta})+
{\mathfrak P}(H_{\gamma, \delta})\right ]\nonumber\\
&-&\frac{1}{1-|\lambda(\alpha, \beta ;\gamma, \delta)|^2}\left [{\mathfrak P}(H_{\alpha,\beta}){\mathfrak P}(H_{\gamma, \delta})+
{\mathfrak P}(H_{\gamma, \delta}){\mathfrak P}(H_{\alpha,\beta})\right ].
\end{eqnarray}
We can now confirm Eq.(\ref{e3}) for this example, using the relations
\begin{eqnarray}
&&{\mathfrak P}(H_{\alpha,\beta}){\mathfrak P}(H_{\gamma, \delta}){\mathfrak P}(H_{\alpha,\beta})=|\lambda(\alpha, \beta ;\gamma, \delta)|^2{\mathfrak P}(H_{\alpha,\beta})\nonumber\\
&&{\mathfrak P}(H_{\gamma, \delta}){\mathfrak P}(H_{\alpha,\beta}){\mathfrak P}(H_{\gamma, \delta})=|\lambda(\alpha, \beta ;\gamma, \delta)|^2{\mathfrak P}(H_{\gamma, \delta}).
\end{eqnarray}

The above example shows that if $H_1\wedge H_2={\cal O}$ then
${\mathfrak D}(H_1, H_2)={\mathfrak P}(H_1\vee H_2)-{\mathfrak P}(H_1)-{\mathfrak P}(H_2)$, is in general non-zero.
The stronger assumption ${\mathfrak P}(H_1){\mathfrak P}(H_2)=0$ implies that $[{\mathfrak P}(H_1),{\mathfrak P}(H_2)]=0$,
and this leads to
${\mathfrak D}(H_1, H_2)=0$.

\section{The Dempster-Shafer theory: multivaluedness and lower and upper probabilities}

Dempster \cite{DS1} studied a multivalued map from a sample space $X$ to another space $\Omega$. He has shown that due to the multivaluedness of the map
Kolmogorov probabilities related to subsets of $X$, become lower and upper probabilities on subsets of $\Omega$.
We summarize very briefly, the basics of the Dempster-Shafer formalism starting with a simple example, which is similar to one in ref.\cite{Za}.

A company has $n=n_1+n_2+n_3$ employees. The age of $n_1$ employees of is known to be under $30$, $n_2$ employees are known to be over $50$, 
and the rest $n_3$ employees are known to be between $25$ and $45$. 
The Dempster multivaluedness is here the fact that our knowledge about the age of each employee, is an interval of values, rather than a single value.
We want to find the probability that a random employee
is under $35$. Let $S$ be the set of employees under $35$.
We are certain that $n_1$ employees belong to $S$, and that $n_2$ employees do not belong to $S$.
The $n_3$ employees belong to the `don't know' category (Dempster \cite{DS2} emphasized the importance of this category).
The Dempster-Shafer theory introduces the lower probability or belief $\ell$, and the upper probability or plausibility $u$, given by
\begin{eqnarray}
\ell=\frac{n_1}{n};\;\;\;\;\;u=\frac{n_1+n_2}{n}.
\end{eqnarray}
They quantify what in everyday language is called `worst case scenario' and `best case scenario'.
We express these concepts in a more formal way.

Given a set $\Omega$, let $A,B$ be elements of the powerset $2^{\Omega}$ (i.e., subsets of $\Omega$).
The lower probability or belief $\ell (A)$, is a monotone function from $2^{\Omega}$ to $[0,1]$, i.e., 
\begin{eqnarray}\label{28A}
A \subseteq B\;\;\rightarrow\;\;\ell (A)\le \ell (B),
\end{eqnarray}
The $1-\ell (\overline A)$ is in general different from $\ell (A)$, and it is the upper probability or plausibility $u(A)$
(where $\overline A=\Omega -A$).
In the above example, ${ \ell } ( \overline A)$ is the lower probability that the age of the average employee is over $35$,
and it is equal to $n_2/n$.

The difference between the upper and lower probabilities, describes the `don't know' case:
\begin{eqnarray}
u(A)-\ell (A)=1-\ell (A)-{ \ell } ( \overline A).
\end{eqnarray}
The upper probability combines the `true' and the `don't know'. 

For Kolmogorov probabilities $1-q({\overline A})=q(A)$ which means that the statement `belongs to $A$' is equivalent with the statement
`does not belong to $\overline A$'. In the Dempster-Shafer theory, due to the `don't know' cases, this is not true in general.

The properties of lower and upper probabilities and also for comparison, of Kolmogorov probabilities, are summarized in table \ref{t1}.
Kolmogorov theory can be viewed as the case where all upper probabilities are equal to the corresponding lower probabilities.

\begin{remark}
From Eq.(\ref{I1}) follows the additivity property of Kolmogorov probabilities
\begin{eqnarray}
A\cap B=\emptyset\;\;\rightarrow \;\;q(A\cup B)=q(A)+q(B).
\end{eqnarray}
This is important for integration. 
Dempster-Shafer probabilities are capacities, i.e., they obey the weaker property of Eq.(\ref{28A}).
A different integration concept, known as Choquet integrals\cite{INT1,INT2}, is applicable to them, but we do not use it in this paper.
\end{remark}

\section{Quantum probabilities as Dempster-Shafer probabilities due to non-commutativity}

We interpret a pair $p(H_1|\rho)$ and $p(H_2|\rho)$ of quantum probabilities, as upper or lower probabilities as follows:
\begin{itemize}
\item
If ${\mathfrak d}(H_1, H_2|\rho)\ge 0$, then  $p(H_1|\rho)$ and $p(H_2|\rho)$ are lower probabilities.
Also, according to Eq.(\ref{32a}), in this case ${\mathfrak d}(H_1^{\bot}, H_2^{\bot}|\rho)\le 0$, and therefore
the  $p(H_1^{\bot}|\rho)$ and $p(H_2^{\bot}|\rho)$ are upper probabilities.

\item
If ${\mathfrak d}(H_1, H_2|\rho)\le 0$, then  $p(H_1|\rho)$ and $p(H_2|\rho)$ are upper probabilities, and 
the  $p(H_1^{\bot}|\rho)$ and $p(H_2^{\bot}|\rho)$ are lower probabilities.

\item
If ${\mathfrak d}(H_1, H_2|\rho)=0$, then  $p(H_1|\rho)$ and $p(H_2|\rho)$ are Kolmogorov probabilities.
Also the $p(H_1^{\bot}|\rho)$ and $p(H_2^{\bot}|\rho)$ are Kolmogorov probabilities.

\item
The characterization lower or upper probabilities refers to a particular pair.
It may be that $p(H_1|\rho)$ is a lower probability when paired with $p(H_2|\rho)$, and upper probability when paired with $p(H_3|\rho)$.
\end{itemize}

It is clear that if $H_1,H_2$ are elements of a Boolean algebra, then  $p(H_1|\rho)$, $p(H_2|\rho)$,
are Kolmogorov probabilities, and a relation analogous to Eq.(\ref{I1}) holds. But for general elements $H_1,H_2$ in ${\cal L}[H(d)]$
the $p(H_1|\rho)$, $p(H_2|\rho)$, are Dempster-Shafer probabilities, and a relation analogous to Eq.(\ref{I1}) does not hold.
Non-zero value of the commutator $[{\mathfrak P}(H_1),{\mathfrak P}(H_2)]$, implies non-zero value of ${\mathfrak D}(H_1, H_2)$,
and this leads to the adoption of Dempster-Shafer probabilities, which do not obey Eq.(\ref{I1}).

We summarize our motivation for using the Dempster-Shafer theory in a quantum context:
\begin{itemize}
\item[(1)]
Kolmogorov probabilities obey Eq.(\ref{I1}), while Dempster-Shafer probabilities might violate it (see table \ref{t1}).
Quantum probabilities violate the analogue of Eq.(\ref{I1}), and this is  directly linked to non-commutativity. Eq.(\ref{e3}) shows that 
the commutator $[{\mathfrak P}(H_1), {\mathfrak P}(H_2)]$
is intimately related to ${\mathfrak D}(H_1, H_2)$, which quantifies deviations from Kolmogorov probabilities.
Therefore in quantum mechanics we need a more general (than Kolmogorov) probability theory, and we propose the Dempster-Shafer theory.

\item[(2)]
Kolmogorov probabilities are intimately connected to Boolean algebras.
Within the ${\cal L}[H(d)]$ there are Boolean algebras, and in these `islands' quantum probabilities behave like Kolmogorov probabilities.
But in the full lattice, we need a more general probability theory, and we propose the 
Dempster-Shafer theory. 

\item[(3)]
The passage from classical to quantum mechanics can be viewed as a type of Dempster multivaluedness.
The product of two classical quantities, becomes a product of two operators, which can be ordered in many ways (as
${\mathfrak P}(H_1){\mathfrak P}(H_2)$ or as ${\mathfrak P}(H_2){\mathfrak P}(H_1)$ or in many other intermediate ways).
This can be interpreted as a type of Dempster multivaluedness, the `spread' of which is quantified with the commutator $[{\mathfrak P}(H_1), {\mathfrak P}(H_2)]$.
Eq.(\ref{e3}) shows that this `spread' is intimately related to ${\mathfrak D}(H_1, H_2)$, which quantifies deviations from Kolmogorov probabilities.
Dempster's `don't know' becomes here `don't know which ordering rule to use'.
So the motivation for introducing an interval of probabilities (from the lower to the upper one), is because there are many products of the operators
${\mathfrak P}(H_1), {\mathfrak P}(H_2)$.
We point out here the analogy with the $Q$-function, Wigner function and $P$-function,  which are part of a continuum of quantities (not probabilities)
related to the ordering of operators (e.g., \cite{A1,A2}).
The lower and upper probabilities is another language for these problems, which might provide a deeper insight to phase space methods in quantum mechanics.

\item[(4)]
The Dempster-Shafer theory uses non-additive probabilities (i.e., it violates Eq.(\ref{I1})).
Non-additive probabilities are used in areas like Game theory or Operations Research to describe coalitions (e.g., the merger of two companies).
In everyday language this non-additivity is described with the expression
`the whole is greater than the sum of its parts'. 
Non-additive probabilities might be better in describing problems like contextuality, than additive (Kolmogorov) probabilities.
A measurement $M$ preformed in conjuction with the measurements $A_1,A_2,..$,
might give a different result from the measurement $M$ preformed in conjuction with the measurements $B_1,B_2,..$\cite{C3}.
Here the $M,A_1,A_2,..$ commute with each other and the $M,B_1,B_2,..$ also commute with each other,
but the $A_1,A_2,..$ might not commute with the $B_1,B_2,..$.
So the measurement $M$ behaves in a different way, when it is in a `coalition' with the $A_1,A_2,..$, than when it is in a coalition with the $B_1,B_2,..$.
In the next section we make a small step in this direction, and we show that Dempster-Shafer probabilities violate Bell inequalities.
The fact that experiment also violates these inequalities, supports the adoption of the Dempster-Shafer theory.
\end{itemize}

\section{Dempster-Shafer probabilities violate Bell inequalities}

The purpose of this section is to show that the proof of Bell inequalities relies on the properties of Kolmogorov
probabilities, and it is not valid
for Dempster-Shafer probabilities. 
The violation of these inequalities in experiments, supports the interpretation of quantum probabilities as Dempster-Shafer probabilities.

\subsection{Boolean algebras within ${\cal L}[H(4)]$}

We consider the subspaces of $H(4)$ shown in table \ref{t2}.
Also ${\cal I}=H(4)$ and ${\cal O}$ is the zero-dimensional subspace that contains only the zero vector.
The 
\begin{eqnarray}
B_A=\{{\cal I},{\cal O}, H_{iA}, H_{iA}^{\bot}\;|\;i=1,...,7\}
\end{eqnarray} 
is a sublattice of ${\cal L}[H(4)]$. 
It is easily seen that $B_A$ is closed under the various logical operations. 
Also any two of its $16$ elements commute ($H_{iA}{\cal C} H_{jA}$). 
The sublattice $B_A$ is a Boolean algebra within ${\cal L}[H(4)]$. Negation is defined as $\neg H_{iA}=H_{iA}^{\bot}$.
Projectors to the spaces in $B_A$, in terms of the projectors in Eq.(\ref{3}), are given in table \ref{t2}.

With any $SU(4)$ transformation, we get another Boolean algebra. 
For later use we consider the transformation 
\begin{eqnarray}
U_B={\bf 1}_2\otimes U(a,b)
\end{eqnarray}
and we get the Boolean algebra $B_B$ shown in table \ref{t2}.
We use the notation $H_{iB}$ for the Hilbert space that contains the states $U_B\ket{s}$ where $\ket{s}$ is a state in $H_{iA}$.
In particular we note that $H_{5A}=H_{5B}$ and $H_{5A}^{\bot}=H_{5B}^{\bot}$.
The Hilbert spaces in
\begin{eqnarray}
B_A\cap B_B=\{{\cal I},{\cal O},H_{5A}, H_{5A}^{\bot}\}
\end{eqnarray}  
commute with all spaces in both Boolean algebras $B_A,B_B$.
But the spaces in $B_A-(B_A\cap B_B)$, in general do not commute with the spaces in $B_B-(B_A\cap B_B)$.
Similar statement can be made for the projectors to these spaces.

For later use we also consider the transformations 
\begin{eqnarray}
U_C=U(a,b)\otimes {\bf 1}_2;\;\;\;\;\;
U_D=U(a,b)\otimes U(a,b)
\end{eqnarray}
and we get the Boolean algebra $B_C, B_D$ correspondingly.
We use notation analogous to the above, for the Hilbert spaces in these Boolean algebras.

\subsection{CHSH inequalities for a system of two spin $1/2$ particles}

A logical derivation of Bell-like inequalities for the case of Boolean variables, has been presented in \cite{C12}
and we have generalized it for Heyting algebras in \cite{vou}.
We present briefly the derivation for Boolean variables, in order to emphasize the crucial role of Boole's inequality,
which is valid for Kolmogorov probabilities, but might not be valid for Dempster-Shafer probabilities.
\begin{proposition}\label{723}
Let  $H_1,...,H_n$ be elements of ${\cal L}[H(d)]$ such that $H_1\wedge ...\wedge H_n={\cal O}$.
If $p(H_i|\rho )$ have the properties of Kolmogorov probabilities, then
\begin{eqnarray}\label{77}
\sum _{i=1}^n p(H_i|\rho )\le n -1
\end{eqnarray}
\end{proposition}
\begin{proof}
We start with the relation
\begin{eqnarray}\label{502}
 H_1^{\bot}\vee...\vee H_n^{\bot}=(H_1\wedge...\wedge H_n)^{\bot}={\cal I}.
\end{eqnarray}
Therefore
\begin{eqnarray}\label{B1}
p(H_1^{\bot}\vee...\vee H_n^{\bot}|\rho)=1.
\end{eqnarray}
For Kolmogorov probabilities, we can use Boole's inequality (see table \ref{t1}) to get 
\begin{eqnarray}\label{B2}
p(H_1^{\bot}|\rho)+...+p(H_n^{\bot}|\rho)\ge 1.
\end{eqnarray}
We then use the relation
\begin{eqnarray}
p(H_i^{\bot}|\rho)=1-p( H_i|\rho),
\end{eqnarray}
and we get Eq.(\ref{77}).
\end{proof}
\begin{remark}
The step from Eq.(\ref{B1}) to Eq.(\ref{B2}) might not be valid for Dempster-Shafer probabilities, and the inequality in Eq.(\ref{77}) might be violated. Below we give an example of this.
\end{remark}

We consider a system of two spin $1/2$ particles, described with the Hilbert space $H(4)=H(2)\otimes H(2)$.
The system is in the state
\begin{eqnarray}\label{state}
\ket{s}=\frac{1}{\sqrt{2}}\left [\ket{\frac{1}{2},\frac{1}{2}}\otimes \ket{\frac{1}{2},\frac{1}{2}}+\ket{\frac{1}{2},-\frac{1}{2}}
\otimes \ket{\frac{1}{2},-\frac{1}{2}}\right ]
=\frac{1}{\sqrt{2}}(1,0,0,1)^T
\end{eqnarray}
We consider the measurements described with the operators
\begin{eqnarray}\label{MA}
A&=&S_x\otimes S_x=\frac{1}{4}[\Pi(x, 1)\otimes \Pi(x ,1)+\Pi(x ,0)\otimes \Pi(x, 0)]\nonumber\\
&-&\frac{1}{4}[\Pi(x, 1)\otimes \Pi(x, 0)+\Pi(x ,0)\otimes \Pi(x, 1)]
=\frac{1}{4}[{\mathfrak P}(H_{1A})+{\mathfrak P}(H_{4A})]-\frac{1}{4}[{\mathfrak P}(H_{2A})+{\mathfrak P}(H_{3A})]
\end{eqnarray}
\begin{eqnarray}\label{MB}
B&=&S_x\otimes S_{a,b}=\frac{1}{4}[\Pi(x, 1)\otimes \Pi(a,b;1)+\Pi(x ,0)\otimes \Pi(a,b; 0)]\nonumber\\
&-&\frac{1}{4}[\Pi(x ,1)\otimes \Pi(a,b;0)+\Pi(x, 0)\otimes \Pi(a,b;1)]
=\frac{1}{4}[{\mathfrak P}(H_{1B})+{\mathfrak P}(H_{4B})]-\frac{1}{4}[{\mathfrak P}(H_{2B})+{\mathfrak P}(H_{3B})]
\end{eqnarray}
\begin{eqnarray}\label{MC}
C&=&S_{a,b}\otimes S_x=\frac{1}{4}[\Pi(a,b;1)\otimes \Pi(x ,1)+\Pi(a,b; 0)\otimes \Pi(x, 0)]\nonumber\\
&-&\frac{1}{4}[\Pi(a,b; 1)\otimes \Pi(x, 0)+\Pi(a,b;0)\otimes \Pi(x ,1)]
=\frac{1}{4}[{\mathfrak P}(H_{1C})+{\mathfrak P}(H_{4C})]-\frac{1}{4}[{\mathfrak P}(H_{2C})+{\mathfrak P}(H_{3C})]
\end{eqnarray}
\begin{eqnarray}\label{MD}
D&=&S_{a,b}\otimes S_{a,b}=\frac{1}{4}[\Pi(a,b;1)\otimes \Pi(a,b;1)+\Pi(a,b; 0)\otimes \Pi(a,b;0)]\nonumber\\
&-&\frac{1}{4}[\Pi (a,b; 1)\otimes \Pi(a,b;0)+\Pi(a,b;0)\otimes \Pi(a,b;1)]\nonumber\\
&=&\frac{1}{4}[{\mathfrak P}(H_{1D})+{\mathfrak P}(H_{4D})]-\frac{1}{4}[{\mathfrak P}(H_{2D})+{\mathfrak P}(H_{3D})]
\end{eqnarray}
Each of these measurements gives one of the outcomes $(1,1)$, $(0,1)$, $(1,0)$, $(0,0)$. 
We use the notation $p(A;1,1)$ for the probability that measurement $A$ will give the outcome $(1,1)$.
From table \ref{t1} it is seen that this is also the probability $p(H_{1A}|\ket{s})$:
\begin{eqnarray}
p(A;1,1)=\bra {s}\Pi(x, 1)\otimes \Pi(x ,1)\ket {s}=\bra {s}{\mathfrak P}(H_{1A})\ket {s}=p(H_{1A}|\ket{s})
\end{eqnarray}
More generally
\begin{eqnarray}
&&p(i;1,1)=p(H_{1i}|\ket{s});\;\;\;
p(i;0,1)=p(H_{3i}|\ket{s});\;\;\;
p(i;1,0)=p(H_{2i}|\ket{s});\;\;\;
p(i;0,0)=p(H_{4i}|\ket{s})\nonumber\\
&&i=A,B,C,D.
\end{eqnarray}
The values of the probabilities are shown in table \ref{t4}, where
\begin{eqnarray}\label{proba}
\kappa=\frac{1}{2}(a_R^2+b_I^2);\;\;\;\;\lambda=\frac{1}{8}\left \{|a+b|^4+|a-b|^4+(a^2-b^2)^2+[(a^*)^2-(b^*)^2]^2\right \}.
\end{eqnarray}
Here the indices $R$,$I$ indicate the real and imaginary part, correspondingly.

We next prove that
\begin{eqnarray}\label{756}
[{\mathfrak P}(H_{1A})+{\mathfrak P}(H_{4A})][{\mathfrak P}(H_{1B})+{\mathfrak P}(H_{4B})][{\mathfrak P}(H_{1C})+{\mathfrak P}(H_{4C})]
[{\mathfrak P}(H_{2D})+{\mathfrak P}(H_{3D})]=0.
\end{eqnarray}
Each of the $16$ terms in this product is easily seen to be equal to $0$. For example
\begin{eqnarray}
{\mathfrak P}(H_{1A}){\mathfrak P}(H_{1B}){\mathfrak P}(H_{1C}){\mathfrak P}(H_{2D})&=&
[\Pi(x ,1)\otimes \Pi(x; 1)][\Pi(x ,1)\otimes \Pi(a,b; 1)]\nonumber\\&\times &[\Pi(a,b; 1)\otimes \Pi(x ,1)]
[\Pi(a,b; 1)\otimes \Pi(a,b; 0)]=0
\end{eqnarray}
We then point out that ${\mathfrak P}(H_{1A}){\mathfrak P}(H_{4A})=0$ and therefore ${\mathfrak P}(H_{1A})+{\mathfrak P}(H_{4A})={\mathfrak P}(H_{1A} \vee H_{4A})$.
Similar comment can be made for the other factors in Eq(\ref{756}), and therefore it can be written as
\begin{eqnarray}
{\mathfrak P}(H_{1A} \vee H_{4A}){\mathfrak P}(H_{1B} \vee H_{4B}){\mathfrak P}(H_{1C} \vee H_{4C})
{\mathfrak P}(H_{2D} \vee H_{3D})=0.
\end{eqnarray}
From this follows that
\begin{eqnarray}\label{678}
[H_{1A} \vee H_{4A}]\wedge [H_{1B} \vee H_{4B}]\wedge [H_{1C} \vee H_{4C}]\wedge [H_{2D} \vee H_{3D}]={\cal O}.
\end{eqnarray}
Then proposition \ref{723}, gives the following CHSH inequality:
\begin{eqnarray}\label{987}
p(H_{1A} \vee H_{4A}|\ket{s})+p(H_{1B} \vee H_{4B}|\ket{s})+p(H_{1C} \vee H_{4C}|\ket{s})+p(H_{2D} \vee H_{3D}|\ket{s})\le 3.
\end{eqnarray}
The fact that  ${\mathfrak P}(H_{1A})+{\mathfrak P}(H_{4A})={\mathfrak P}(H_{1A} \vee H_{4A})$ implies that
$p(H_{1A} \vee H_{4A}|\ket{s})=p(H_{1A}|\ket{s})+p(H_{4A}|\ket{s})$, and similarly for the other terms in Eq.(\ref{987}), 
which can be written as
\begin{eqnarray}
&&p(H_{1A}|\ket{s})+p(H_{4A}|\ket{s})+p(H_{1B}|\ket{s})+p(H_{4B}|\ket{s})+p(H_{1C}|\ket{s})\nonumber\\&&+p(H_{4C}|\ket{s})+p(H_{2D}|\ket{s})+p(H_{3D}|\ket{s})\le 3.
\end{eqnarray}
We insert the values from table \ref{t4} and we get
\begin{eqnarray}\label{951}
1+4\kappa+(1-2\lambda) \le 3
\end{eqnarray}
As an example, we consider the case with $a=\exp (i\theta)$ and $b=0$. Then $\kappa =\frac{1}{2}(\cos \theta)^2$ and $\lambda=\frac{1}{8}(2+2 \cos 4\theta)$.
Therefore the inequality of Eq.(\ref{951}) becomes
\begin{eqnarray}
2(\cos \theta)^2-\frac{1}{2}(1+\cos 4\theta)\le 1.
\end{eqnarray}
For $\theta=\pi/8$ this inequality is violated.

\begin{remark}
We have emphasized that the derivation of the inequality in Eq.(\ref{77}), 
which is used in Eq.(\ref{987}), relies on the properties of Kolmogorov probabilities,
and in particular on Boole's inequality.
We now show explicitly (i) that the probabilities in Eq.(\ref{987}) are not Kolmogorov probabilities,
and (ii) that Boole's inequality is violated. 
\begin{itemize}
\item[(i)]
We consider the $p(H_{1A} \vee H_{4A}|\ket{s})$ and $p(H_{1B} \vee H_{4B}|\ket{s})$.
We use the relations
\begin{eqnarray}
H_{1A} \vee H_{4A}=H_{7A};\;\;\;\;\;H_{1B} \vee H_{4B}=H_{7B}
\end{eqnarray}
and we find that the commutator of the corresponding projectors, is
\begin{eqnarray}
&&[{\mathfrak P} (H_{1A} \vee H_{4A}),{\mathfrak P} (H_{1B} \vee H_{4B})]=[{\mathfrak P} (H_{7A}),{\mathfrak P} (H_{7B})]\nonumber\\&&=
\Pi(x, 1)\otimes [\Pi(x, 1), \Pi(a,b; 1)]+\Pi(x, 0)\otimes [\Pi(x, 0),\Pi(a,b; 0)].
\end{eqnarray}
According to Eq.(\ref{e3}) in this case ${\mathfrak D}(H_{1A} \vee H_{4A}, H_{1B} \vee H_{4B})\ne 0$ and the corresponding probabilities 
are not Kolmogorov probabilities. 

\item[(ii)]
From Eq.(\ref{678}) it follows that
\begin{eqnarray}
[H_{1A} \vee H_{4A}]^{\bot}\vee [H_{1B} \vee H_{4B}]^{\bot}\vee[H_{1C} \vee H_{4C}]^{\bot}\vee[H_{2D} \vee H_{3D}]^{\bot}={\cal I}.
\end{eqnarray}
Therefore the corresponding probability is 
\begin{eqnarray}\label {B10}
p\left ([H_{1A} \vee H_{4A}]^{\bot}\vee [H_{1B} \vee H_{4B}]^{\bot}\vee[H_{1C} \vee H_{4C}]^{\bot}\vee[H_{2D} \vee H_{3D}]^{\bot}\right )=1.
\end{eqnarray}
On the other hand from table \ref{t4}, it is seen that
\begin{eqnarray}
&&p\left ([H_{1A} \vee H_{4A}]^{\bot}\right )=1-p[H_{1A} \vee H_{4A}]=0\nonumber\\
&&p\left ([H_{1B} \vee H_{4B}]^{\bot}\right )=1-p[H_{1B} \vee H_{4B}]=1-2\kappa\nonumber\\
&&p\left ([H_{1C} \vee H_{4C}]^{\bot}\right )=1-p[H_{1C} \vee H_{4C}]=1-2\kappa\nonumber\\
&&p\left ([H_{2D} \vee H_{3D}]^{\bot}\right )=1-p[H_{2D} \vee H_{3D}]=1-(1-2\lambda)
\end{eqnarray}
and therefore
\begin{eqnarray}\label{B20}
p\left ([H_{1A} \vee H_{4A}]^{\bot}\right )+
p\left ([H_{1B} \vee H_{4B}]^{\bot}\right )+
p\left ([H_{1C} \vee H_{4C}]^{\bot}\right )+
p\left ([H_{2D} \vee H_{3D}]^{\bot}\right )=2-4\kappa +2\lambda.
\end{eqnarray}
We have seen earlier an example where $2-4\kappa +2\lambda$ is smaller than $1$ and therefore
the sum of probabilities in Eq.(\ref{B20}), is smaller than the probability in Eq.(\ref{B10}),
which shows explicitly that Boole's inequality is violated. 
\end{itemize} 
\end{remark}

\section{Discussion}

An important property of Kolmogorov probabilities is that $\delta (A,B)=0$ (Eq.(\ref{I1})).
We have introduced the operator ${\mathfrak D}(H_1,H_2)$ of Eq.(\ref{32}), which is analogous to
$\delta (A,B)$ and we have shown that it is related to the commutator $[{\mathfrak P} (H_1),{\mathfrak P} (H_2)]$ as in Eq.(\ref{e3}).
This shows a direct link between commutativity and the Kolmogorov property $\delta (A,B)=0$.
If $H_1,H_2$ belong to the same Boolean subalgebra of the orthomodular lattice ${\cal L}[H(d)]$, 
then ${\mathfrak D}(H_1,H_2)=[{\mathfrak P} (H_1),{\mathfrak P} (H_2)]=0$.
In this case the corresponding probabilities $p(H_1|\rho), p(H_2|\rho)$, obey
Eq.(\ref{I1}), and they are Kolmogorov probabilities.
But in general, the ${\mathfrak D}(H_1,H_2)$ and $[{\mathfrak P} (H_1),{\mathfrak P} (H_2)]$ are non-zero,
and then the $p(H_1|\rho), p(H_2|\rho)$, do not obey 
Eq.(\ref{I1}), and they are not Kolmogorov probabilities.

The Dempster-Shafer theory is designed for `real world' data which have multivaluedness and contradictions.
Dempster-Shafer probabilities have the properties shown in table \ref{t1}, which fit very well with the requirements of quantum probabilities.
In particular they violate Eq.(\ref{I1}). The difference between upper and lower probabilities, is Dempster's `don't know' cases,
which here is related to ambiguity in the ordering of quantum mechanical operators. In the semiclassical limit, the commutators go to zero,
and the ${\mathfrak D}(H_1,H_2)$ goes to zero, and lower probabilities become equal to upper probabilities, i.e., they become Kolmogorov probabilities. 

As an application we have considered CHSH inequalities and we have stressed that their proof relies on the properties of Kolmogorov probabilities, 
and it is not valid for Dempster-Shafer probabilities. Their violation in experiments, supports the use of the Dempster-Shafer theory for quantum probabilities.

In this work we have considered systems with finite Hilbert space, and then ${\cal L}[H(d)]$ is a modular lattice.
This was used, for example, in the first part of proposition \ref{pro}.  
However, our main theme of using Dempster-Shafer probabilities as quantum probabilities, 
could also be applied to systems with infinite-dimensional Hilbert spaces.
In this case the lattice is orthomodular but not modular, not every subspace is closed, etc.
So there are extra mathematical questions, which need to be considered.

\newpage
\begin{table}
\caption{Properties of the lower and upper probabilities in the Dempster-Shafer theory, and also of the Kolmogorov probabilities.
$A,B$ are subsets of $\Omega$.}
\centering
\begin{tabular}{|c|c|}\hline
{\rm \bf lower probabilities} $\ell (A)$ &{\rm \bf upper probabilities} $u(A)=1-\ell (\overline A)$\\\hline
$A \subseteq B\;\;\rightarrow\;\;\ell (A)\le \ell (B)$&$A \subseteq B\;\;\rightarrow\;\;u (A)\le u (B)$\\\hline
$\ell (\emptyset)=0;\;\;\;\;\ell (\Omega)=1$&$u (\emptyset)=0;\;\;\;\;u (\Omega)=1$\\\hline
$\ell (A\cup B)-\ell (A)-\ell (B)+\ell (A\cap B)\ge 0$&$u (A\cup B)-u (A)-u (B)+u (A\cap B)\le 0$\\\hline
$\ell (\overline A)+\ell (A)\le 1$&$u (\overline A)+u (A)\ge 1$\\\hline
$\ell (A)+\ell (B)-\ell (A\cup B)$ &$u(A)+u(B)-u(A\cup B)\ge 0$\\
{\rm might be negative}&{\rm Boole's inequality}\\\hline\hline
{\rm \bf Kolmogorov probabilities} $q(A)$&\\\hline
$A \subseteq B\;\;\rightarrow\;\;q (A)\le q (B)$&\\\hline
$q (\emptyset)=0;\;\;\;\;q (\Omega)=1$&\\\hline
$q (A\cup B)-q (A)-q(B)+q(A\cap B)= 0$&\\\hline
$q (\overline A)+q (A)= 1$&\\\hline
$q (A)+q (B)-q (A\cup B)\ge 0$&\\
{\rm Boole's inequality}&\\\hline
\end{tabular}\label{t1}

\vspace{0.7cm}

\caption{The Hilbert spaces in the Boolean algebra $B_A$ within ${\cal L}[H(4)]$, and the corresponding projectors.
The ${\cal I}=H(4)$ and ${\cal O}$ also belong to $B_A$. The general  vector that belongs to each of these spaces is shown.}
\centering
\begin{tabular}{|c|||c|}\hline
$H_{1A}=\left \{a (1,1,1,1)^T\right \}$&${\mathfrak P}(H_{1A})=\Pi(x ,1)\otimes \Pi(x, 1)$\\\hline
$H_{1A}^{\bot}=\left \{(a,b,c, -a-b-c)^T\right \}$&${\mathfrak P}(H_{1A}^{\bot})={\bf 1}_4-\Pi(x, 1)\otimes \Pi(x ,1)$\\\hline
$H_{2A}=\left \{a (1,-1,1,-1)^T\right \}$&${\mathfrak P}(H_{2A})=\Pi(x, 1)\otimes \Pi(x, 0)$\\\hline
$H_{2A}^{\bot}=\left \{(a,b,c,a-b+c )^T\right \}$&${\mathfrak P}(H_{2A}^{\bot})={\bf 1}_4-\Pi(x, 1)\otimes \Pi(x, 0)$\\\hline
$H_{3A}=\left \{a(1,1,-1,-1)^T\right \}$&${\mathfrak P}(H_{3A})=\Pi(x, 0)\otimes \Pi(x, 1)$\\\hline
$H_{3A}^{\bot}=\left \{(a,b,c,a+b-c)^T\right \}$&${\mathfrak P}(H_{3A}^{\bot})={\bf 1}_4-\Pi(x, 0)\otimes \Pi(x, 1)$\\\hline
$H_{4A}=\left \{a (1,-1,-1,1)^T\right \}$&${\mathfrak P}(H_{4A})=\Pi(x, 0)\otimes \Pi(x, 0)$\\\hline
$H_{4A}^{\bot}=\left \{(a,b,c,-a+b+c)^T\right \}$&${\mathfrak P}(H_{4A}^{\bot})={\bf 1}_4-\Pi(x, 0)\otimes \Pi (x, 0)$\\\hline
$H_{5A}=\left \{(a,b,a,b)^T\right \}$&${\mathfrak P}(H_{5A})=\Pi(x, 1)\otimes {\bf 1}_2$\\\hline
$H_{5A}^{\bot}=\left \{(a,b,-a,-b)^T\right \}$&${\mathfrak P}(H_{5A}^{\bot})={\bf 1}_4-\Pi(x, 1)\otimes {\bf 1}_2$\\\hline
$H_{6A}=\left \{(a,a,b,b)^T\right \}$&${\mathfrak P}(H_{6A})={\bf 1}_2\otimes \Pi(x, 1)$\\\hline
$H_{6A}^{\bot}=\left \{(a,-a,b,-b)^T\right \}$&${\mathfrak P}(H_{6A}^{\bot})={\bf 1}_4-{\bf 1}_2\otimes \Pi(x, 1)$\\\hline
$H_{7A}=\left \{(a,b,b,a)^T\right \}$&${\mathfrak P}(H_{7A})=\Pi(x, 1)\otimes \Pi(x, 1)+\Pi(x, 0)\otimes \Pi(x, 0)$\\\hline
$H_{7A}^{\bot}=\left \{(a,b,-b,-a)^T\right \}$&${\mathfrak P}(H_{7A}^{\bot})={\bf 1}_4-\Pi(x, 1)\otimes \Pi(x, 1)-
\Pi(x, 0)\otimes \Pi(x, 0)$\\\hline
\end{tabular}\label{t2}
\vspace{0.7cm}
\end{table}
\newpage
\begin{table}

\caption{The Hilbert spaces in the Boolean algebra $B_B$ within ${\cal L}[H(4)]$, and the corresponding projectors.
The ${\cal I}=H(4)$ and ${\cal O}$ also belong to $B_B$. }
\centering
\begin{tabular}{|c|||c|}\hline
$H_{1B}$&${\mathfrak P}(H_{1B})=\Pi(x ,1)\otimes \Pi(a,b; 1)$\\\hline
$H_{1B}^{\bot}$&${\mathfrak P}(H_{1B}^{\bot})={\bf 1}_4-\Pi(x, 1)\otimes \Pi(a,b; 1)$\\\hline
$H_{2B}$&${\mathfrak P}(H_{2B})=\Pi(x, 1)\otimes \Pi(a,b; 0)$\\\hline
$H_{2B}^{\bot}$&${\mathfrak P}(H_{2B}^{\bot})={\bf 1}_4-\Pi(x, 1)\otimes \Pi(a,b; 0)$\\\hline
$H_{3B}$&${\mathfrak P}(H_{3B})=\Pi(x, 0)\otimes \Pi(a,b; 1)$\\\hline
$H_{3B}^{\bot}$&${\mathfrak P}(H_{3B}^{\bot})={\bf 1}_4-\Pi(x, 0)\otimes \Pi(a,b; 1)$\\\hline
$H_{4B}$&${\mathfrak P}(H_{4B})=\Pi(x, 0)\otimes \Pi(a,b; 0)$\\\hline
$H_{4B}^{\bot}$&${\mathfrak P}(H_{4B}^{\bot})={\bf 1}_4-\Pi(x, 0)\otimes \Pi (a,b; 0)$\\\hline
$H_{5B}$&${\mathfrak P}(H_{5B})=\Pi(x, 1)\otimes {\bf 1}_2$\\\hline
$H_{5B}^{\bot}$&${\mathfrak P}(H_{5B}^{\bot})={\bf 1}_4-\Pi(x, 1)\otimes {\bf 1}_2$\\\hline
$H_{6B}$&${\mathfrak P}(H_{6B})={\bf 1}_2\otimes \Pi(a,b; 1)$\\\hline
$H_{6B}^{\bot}$&${\mathfrak P}(H_{6B}^{\bot})={\bf 1}_4-{\bf 1}_2\otimes \Pi(a,b; 1)$\\\hline
$H_{7B}$&${\mathfrak P}(H_{7B})=\Pi(x, 1)\otimes \Pi(a,b; 1)+\Pi(x, 0)\otimes \Pi(a,b; 0)$\\\hline
$H_{7B}^{\bot}$&${\mathfrak P}(H_{7B}^{\bot})={\bf 1}_4-\Pi(x, 1)\otimes \Pi(a,b; 1)-\Pi(x, 0)\otimes \Pi(a,b; 0)$\\\hline
\end{tabular}\label{t3}

\caption{Probabilities for the outcome of the measurements $A,B,C,D$, on a system of two spin $1/2$ particles, 
in the state $\ket{s}$ of Eq.(\ref{state})
The values of $\kappa$, $\lambda$ are given in Eq.(\ref{proba}).}
\centering
\begin{tabular}{|c||c|c|c|c|}\hline
& $p(i;1,1)=p(H_{1i}|\ket{s})$& $p(i;0,1)=p(H_{3i}|\ket{s})$&$p(i;1,0)=p(H_{2i}|\ket{s})$&$p(i;0,0)=p(H_{4i}|\ket{s})$\\ \hline\hline
$i=A$&$0.5$& $0$& $0$& $0.5$ \\ \hline
$i=B$&$\kappa$& $0.5-\kappa$& $0.5-\kappa$ & $\kappa$\\ \hline
$i=C$&$\kappa$& $0.5-\kappa$& $0.5-\kappa$ & $\kappa$\\ \hline
$i=D$&$\lambda$& $0.5-\lambda$& $0.5-\lambda$ & $\lambda$\\ \hline
\end{tabular}\label{t4}
\vspace{0.7cm}
\end{table}

\end{document}